\documentclass[conference,onecolumn]{IEEEtran}

\usepackage{algorithm}
\usepackage{algpseudocode}

\usepackage{psfrag}

\usepackage{url}      
\usepackage{pgfplots}
\pgfplotsset{width=7cm,compat=1.3}

\usepackage{psfrag}
\usepackage{makecell}

\usepackage{mathtools}

\usepackage{mathtools}

\usepackage{xfrac}

\usepackage{multicol}

\usepackage {tikz}
\usetikzlibrary {positioning}

\usepackage{latexsym}
\usepackage{pictex}
\usepackage{fancyvrb}
\usepackage{fancyhdr}
\usepackage{color}
\usepackage{subfig}
\usepackage{rawfonts}
\usepackage{graphics}
\usepackage{graphicx}
\usepackage{amssymb,amsmath,amsthm,amsfonts}
\usepackage{ifthen}
\usepackage{bbm}
\usepackage{cite}
\usepackage{enumerate}
\usepackage{verbatim}
\usepackage{balance}

\newcommand{\nop}[1]{} 
\newcommand{\shorten}[1]{}
\newtheorem{proposition}{Proposition}
\newtheorem{theorem}{Theorem}
\newtheorem{definition}{Definition}

\newtheorem{remark}{Remark}
\newtheorem{corollary}{Corollary}
\newtheorem{example}{Example}

\newcommand{\signed}%
    {{\unskip\nobreak\hfill\penalty50
      \hskip2em\hbox{}\nobreak\hfil $\blacksquare$
      \parfillskip=0pt \finalhyphendemerits=0 \par}}

\begin{document}

\title{Storage-Repair Bandwidth Trade-off for Wireless Caching with Partial Failure and Broadcast Repair}

\author{
	\IEEEauthorblockN{Nitish Mital\IEEEauthorrefmark{1}, Katina Kralevska\IEEEauthorrefmark{2}, Deniz G\"{u}nd\"{u}z\IEEEauthorrefmark{1} and Cong Ling\IEEEauthorrefmark{1}\\ 
	\IEEEauthorblockA{\IEEEauthorrefmark{1}Department of Electrical Electronics Engineering, Imperial College London}
	\IEEEauthorblockA{\IEEEauthorrefmark{2}Dep. of Information Security and Communication Technology, NTNU, Norwegian University of Science and Technology}
	Email: \{n.mital,d.gunduz,c.ling\}@imperial.ac.uk, katinak@ntnu.no}
}

\maketitle

\begin{abstract}	
Repair of multiple partially failed cache nodes is studied in a distributed wireless content caching system, where $r$ out of a total of $n$ cache nodes lose part of their cached data. Broadcast repair of failed cache contents at the network edge is studied; that is, the surviving cache nodes transmit broadcast messages to the failed ones, which are then used, together with the surviving data in their local cache memories, to recover the lost content. The trade-off between the storage capacity and the repair bandwidth is derived. It is shown that utilizing the broadcast nature of the wireless medium and the surviving cache contents at partially failed nodes significantly reduces the required repair bandwidth per node. 
\end{abstract}

%
\IEEEpeerreviewmaketitle

\section{Introduction}
Caching popular contents closer to end-users, particularly in the available storage space at the wireless network edge, is attracting a lot of attention in the recent years as a promising method to alleviate the increasing burden on the backhaul links of wireless access points, e.g., small cell base stations, and to improve the quality of service, particularly by reducing the latency \cite{6495773, 8114221}, or energy consumption \cite{7438743}. The literature on distributed coded caching systems focuses mostly on the code design or resource allocation for efficient storage of popular contents, assuming reliable cache nodes. However, storage devices are often unreliable and prone to failures; thus, efficient repair  techniques that guarantee continuous data availability are essential for a successful implementation of distributed caching and content delivery techniques in practice. 

Maximum distance separable (MDS) codes are typically used for distributed caching of contents at multiple access points \cite{6495773, 8114221, DBLP:journals/corr/abs-1712-00649}. MDS codes provide flexibility for storage so that users with different connectivity or mobility patterns can download a file from only a subset of the access points. In particular, an $(n, k)$ MDS code encodes a file of size $M$ bits by splitting it into $k$ equal-size fragments and encoding them into $n$ fragments which are stored at $n$ cache nodes. The original file can be recovered by accessing any $k$ out of $n$ fragments from $k$ distinct access points. When some nodes partially or fully fail, their cache contents need to be regenerated to be able to continue serving users. An important objective of edge caching in wireless networks is to reduce the backhaul link loads; therefore, we will consider \textit{cache recovery at the edge}; that is, rather than updating the failed cache contents from a central server through backhaul links, the failed cache contents are regenerated with the help of surviving cache nodes. The total amount of data transferred from the surviving nodes to repair the failed nodes is called the \textit{repair bandwidth}. Traditional MDS codes have high storage efficiency, but their repair bandwidth is large \cite{5550492}. The data of one node is repaired by accessing and transferring data from $k$ nodes, i.e., by recovering the whole content library.

Dimakis et al. showed in  \cite{5550492} that there is a fundamental trade-off between the storage and repair bandwidth by mapping the repair problem in a distributed storage system to a multicasting problem over an information flow graph. The analysis focuses on a single node repair; that is, losing one of the nodes triggers the repair process. Regenerating codes achieve any point on the optimal trade-off curve, while minimum-storage regenerating (MSR) codes and minimum-bandwidth regenerating (MBR) codes operate on the two extremes of this trade-off curve.


It was observed in \cite{5402494} that multiple node repair; that is, the repair process starts only after $r$ nodes fail, is more efficient in terms of the repair bandwidth per node, compared to repairing each node as it fails. Additionally, we will consider the \textit{partial repair} problem, studied in \cite{7000553}, in which a cache node loses only a part of its contents, and the remainder of the cache contents should be used along with the transmissions from the surviving nodes to recover its original content, thus further reducing the repair bandwidth. In \cite{5978920} and \cite{6565355}, the authors introduce cooperative regenerating codes, which repair multiple failures cooperatively by allowing each of the $r$ nodes being repaired to collect data from the $n-r$ non-failed nodes, and then to cooperate with the other $r-1$ nodes being repaired. Instead, similarly to \cite{7000553}, we will consider broadcast repair; that is, transmissions from each node are received in an error-free manner by all the other nodes. In summary, we will study the broadcast repair of partially failed cache nodes. The storage-repair bandwidth trade-off for the repair of multiple fully failed nodes is investigated in \cite{7459908}.

In \cite{7000553}, the authors derive a lower bound on the number of packet transmissions at the MSR point for error-free partial broadcast repair, and provide an explicit code construction for a special case. The information flow graph construction in \cite{7000553} does not capture the relation between storage capacity per node and the repair bandwidth, thus focusing only on one of the extremal points on the storage-repair bandwidth trade-off curve, corresponding to the MSR point. In this paper, we study the entire optimal trade-off curve. Through specific examples we show that the MSR point from \cite{7000553}, where each node stores $\sfrac{M}{k}$ bits, is not feasible with a finite repair bandwidth for all cases. We also present a code construction for a special case to demonstrate the achievability of the proposed bound using linear network coding.

In \cite{7541450}, the authors consider a centralized model of multiple-node failures that is equivalent to the broadcast repair model in \cite{7459908}. Reference \cite{8277979} investigates the storage-repair bandwidth trade-off for clustered storage networks, where multiple nodes within a cluster fail. This is close to the partial failure model that we consider since each cluster could model multiple memory units within a node, and multiple memory units failing is equivalent to partial failure in a node. However, we consider partial failures at multiple nodes, i.e., multiple clusters, and broadcast transmissions from the non-failed nodes.

\textbf{\textit{Notations}}. For two integers $i<j$, we denote the set $\{i,i+1,\ldots, j \}$ by $[i:j]$, while the set $[1:j]$ is denoted by $[j]$. Sets are denoted with the calligraphic font. Vectors and matrices are denoted with a bold font. 

\section{System Model} \label{nm}
Consider a wireless caching system where $n$ nodes, each with storage capacity $\alpha$ bits, store a file of size $M$ bits. 
We index these storage nodes by the set $\cal N$ $\triangleq \{ 1, \ldots, n \}$. The nodes are fully connected by a wireless broadcast medium and use orthogonal channels for data transmission.

We consider a scenario in which a portion of the stored bits in the storage nodes is subject to being lost. We refer to these nodes as the \textit{faulty nodes} and to the nodes that do not experience any losses as the \textit{complete nodes}. We assume that the repair occurs in rounds, where a repair round gets initiated when $r$ nodes experience partial failures of $\alpha - \alpha_{1}$ bits, where $\alpha_1 \triangleq \rho \alpha , \rho \in [0,1]$. Thus, a single repair round repairs $r$ faulty nodes. There is no loss during a repair round. During a repair round, the lost bits in the faulty nodes are repaired with the help of transmitted bits from the complete nodes and the remaining bits that have not been lost in each of the faulty nodes. In general, the repair is functional, i.e., the repaired portion may not be the same as the original portion, but it satisfies the same property that any $k$ nodes are sufficient to reconstruct the whole file.

\subsection{Information flow graph}

The repair dynamics of the network can be represented by an information flow graph that evolves in time. See Fig. \ref{fig:flow1} and \ref{fig:flow2} for illustrations. It is a directed acyclic graph consisting of six types of nodes: a single source node $S$, storage nodes $x_{in}^i, x_{mid}^i, x_{out}^i$, helpers $h_i$, and a data collector $DC$. 
Each complete storage node $x^i, i \in  [n]$, is represented by two vertices: an input vertex $x_{in}^i$ and an output vertex $x_{out}^i$ that are connected by a directed edge $x_{in}^i \rightarrow x_{out}^i$ with capacity $\alpha$.
A faulty node is represented by four vertices: an input vertex $x_{in}^i$, an intermediate vertex $x^{i}_{mid}$ that is connected to $x_{in}^i$ by a directed edge $x_{in}^i \rightarrow x_{mid}^i$ of capacity $\alpha$, an output vertex $x_{out}^i$ that is connected to $x_{mid}^i$ by a directed edge $x_{mid}^i \rightarrow x_{out}^i$ of capacity $\alpha_1$, and a failed vertex $x_f^i$ that is connected to $x_{mid}^i$ by a directed edge $x_{mid}^i \rightarrow x_{f}^i$ of capacity $\alpha-\alpha_1$. The failed vertex represents the corrupted portion of the data in the storage node. 

Each vertex in the graph at any given time has two modes, active or inactive, depending on its availability. At the initial time, the source node $S$ is active and it transmits data to $n$ storage nodes such that the $DC$ can retrieve the file from any $k$ nodes. This is modeled by adding an edge from $S$ to all input vertices of the storage nodes, $S \rightarrow x_{in}^i, i \in [n]$, with capacity $\infty$\footnote{Note that adding an edge with capacity $\infty$ means that all the information in the node sending the data is available in the input vertices of the nodes receiving the data.}. From this point onwards, the source node becomes inactive, and the storage nodes become active.

When $r$ nodes experience partial failure of $\alpha - \alpha_1$ bits each, in the $s$-th round, the repair process is triggered and $r$ newcomers join the system. Note that a newcomer represents the corresponding node being repaired. A newcomer $x^i$ where $i = sn+j, j\in [n]$, represents the node $x^j$ after the $s$-th round. The lost data is regenerated at the newcomers by receiving functions of the stored data from the $n-r$ complete storage nodes through the helper nodes. 
The $n-r$ complete storage nodes are connected to the corresponding helper nodes with a directed edge $x_{out}^i \rightarrow h^i$ of capacity $\beta$, which denotes the number of bits broadcasted by $x^i$. Each helper node $h^i$ is connected with infinite capacity links to all the newcomers. This represents the broadcast nature of the transmission medium.
\begin{definition} 
The repair bandwidth $\gamma=(n-r)\beta$ is defined as the total number of bits the complete storage nodes broadcast in a repair round.
\end{definition}
We model a newcomer with two vertices $x_{in}^i$ and $x_{out}^i$ and a directed edge $x_{in}^i \rightarrow x_{out}^i$ with capacity $\alpha$. The newcomer $x^i, i={n+(s-1)r+j}$, uses the $\alpha_1$ bits from the corresponding node being repaired. This is captured in the flow graph by edges with capacity $\alpha_1, x_{mid}^i \rightarrow x_{out}^i$, followed by the edges with infinite capacity between the output vertices of the node being repaired and the newcomers.

A data collector corresponds to a request to reconstruct the file. Data collectors connect to any subset of $k$ active nodes and retrieve all the stored data in these nodes, represented with edges with infinite capacity from the active nodes to the $DC$.

\begin{figure*}
		\centering
		\includegraphics[width=0.7\textwidth]{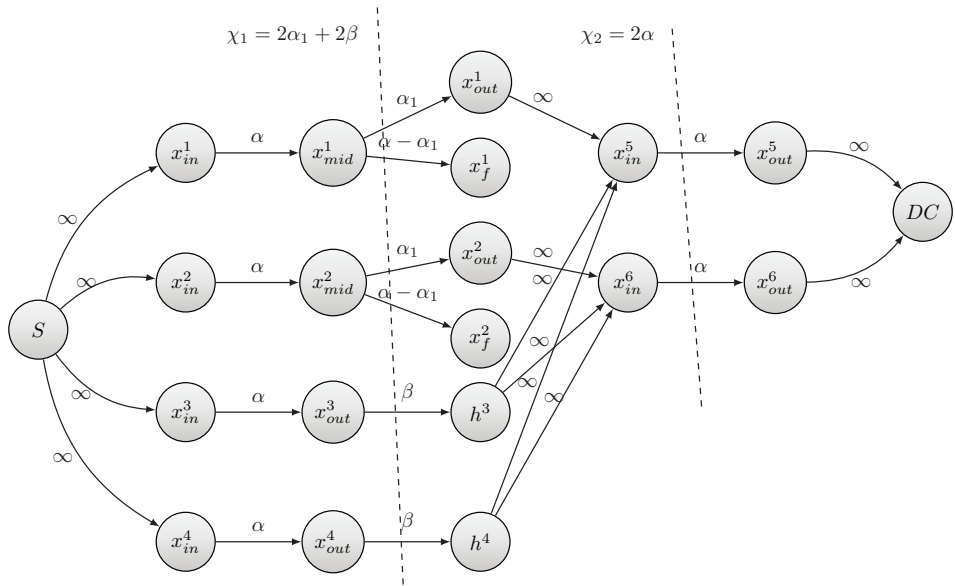} 
        \caption{Information flow graph $\cal G$ with $n=4, k=2, r=2$, one repair round and two cuts $\chi_1$ and $\chi_2$.}
		\label{fig:flow1}
	\end{figure*}

\begin{figure*}
		\centering
		\includegraphics[width=0.7\textwidth]{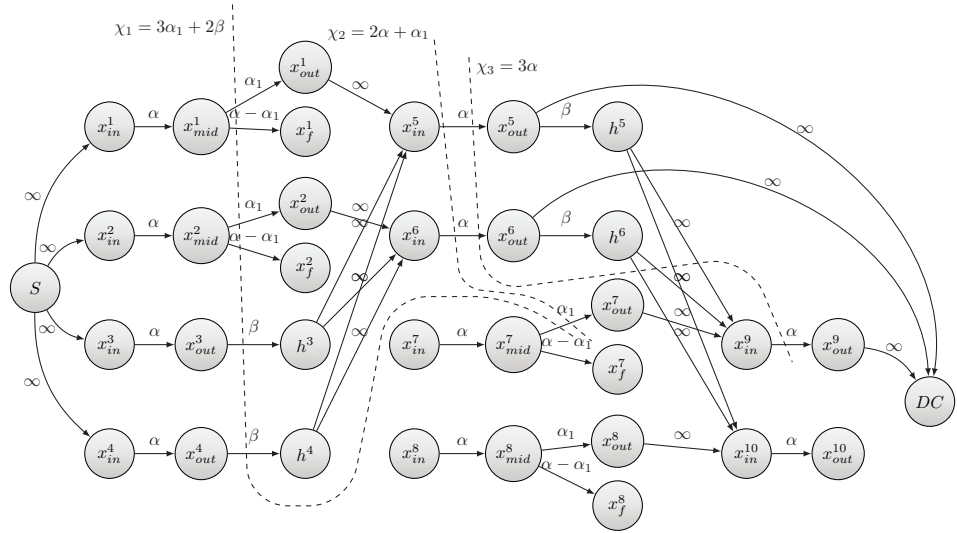} 
        \caption{Information flow graph $\cal G$ with $n=4, k=3, r=2$, two repair rounds and three cuts $\chi_1, \chi_2$ and $\chi_3$.}
		\label{fig:flow2}
	\end{figure*}
A cut in the information flow graph is a subset of edges such that there is no path from the source node $S$ to the data collector $DC$ that does not go through any of the edges in the cut. We define the capacity of a cut as the sum of its edge capacities, and the \textit{min-cut} of a graph as the minimum capacity among all the cuts.
\begin{proposition} \cite{5550492} \label{prop:1}
Consider any given finite information flow graph $\cal G$, with a finite set of data collectors. If the min-cut separating the source from each data collector is larger than or equal to the file size $M$, then there exists a
linear network code defined over a sufficiently large finite field $\mathbb F$
(whose size depends on the graph size) such that all data collectors can recover the original file. Further, randomized network coding guarantees that all collectors can recover the file with probability that can be driven arbitrarily close to $1$ by increasing the field size.
\end{proposition}

Following Proposition \ref{prop:1}, for the information flow graph construction in this paper, we find the minimum cut over all possible failure combinations. We enumerate cuts as $\chi_1, \chi_2, \ldots$ (see Fig. \ref{fig:flow1}, \ref{fig:flow2}). In Section \ref{tradeoff}, we demonstrate how to find the min-cut for a specific example.

\section{Storage-Bandwidth Trade-off for Partial Repair} \label{sec:tradeoff}

Consider the scenario illustrated in Fig. \ref{fig:flow1}, where $n=4, k=2$ and $r=2$. The capacity of cut $\chi_1$ is $2\alpha_1 + 2\beta$, while the capacity of cut $\chi_2$ is $2\alpha$. Then the min-cut is $\min\{2\alpha_1 + 2\beta, 2\alpha\}$. From Proposition \ref{prop:1}, to ensure that the file can be reconstructed by the data collector, $\min\{2\alpha_1 + 2\beta, 2\alpha\} \geq M$.

Next we consider a scenario represented in Fig. \ref{fig:flow2} with $n=4, k=3$ and $r=2$, where two repair rounds are required to determine the min-cut. The number of repair rounds is determined by ensuring that each of the $k$ nodes serving the $DC$ go through at least one repair round, so that all the different capacity edges occur at least once in the path, from $S$ to $DC$, through each node. Therefore, the minimum number of repair rounds required is $\lceil \sfrac{k}{r} \rceil $. The min-cut is then given by $\min \{3\alpha_1 +2\beta, 2\alpha + \alpha_1 \}$, and the sufficient condition for the reconstruction of the file by the $DC$ is $\min \{3\alpha_1 +2\beta, 2\alpha + \alpha_1 \} \geq M$.


For each set of parameters $(n, k, \gamma, \alpha, r, \rho)$, there is a family of information flow graphs, each of which corresponds to a particular evolution of node failures/repairs. We denote this family of directed acyclic graphs by $\cal G$$(n, k, \gamma, \alpha, r, \rho)$. An $(n, k, \gamma, \alpha, r, \rho)$ tuple is feasible, if a code with storage $\alpha$ and repair bandwidth $\gamma$ exists.

\begin{theorem}\label{theorem 1}
For any $\alpha \geq \alpha^{*}(n,k,\gamma, r, \rho)$, the points $(n,k, \gamma, \alpha, r, \rho)$ are feasible, and linear network codes suffice to achieve them. It is information theoretically impossible to achieve points with $\alpha < \alpha^{*}(n,k,\gamma, r, \rho)$. If $r$ divides $k$, the threshold function $\alpha^{*}(n,k,\gamma, r, \rho)$ is given by:
\begin{align}
\alpha^{*}(n,k,\gamma, r, \rho)=\left\{
                \begin{array}{ll}
                  \frac{M}{k} \hspace{15.5mm} \gamma \in \left[ f(0), \infty \right) \\
                  \frac{M - g(i)\gamma}{k-ir(1-\rho)} \hspace{5mm} \gamma \in \left[ f(i), f(i-1)\right] 
                \end{array}
              \right.
\end{align}
where, for $i={1,2,\ldots , \frac{k}{r}-1},$
\begin{align}
f(i)&\triangleq\frac{2M(1-\rho)(n-r)}{(2k-r(i+1)(1-\rho))i+\frac{2k}{r}(n-k)},\\
g(i)&\triangleq \frac{1}{2}\left(2n-2k-r+ir\right)\frac{ir}{n-r}.
\end{align}
\end{theorem}

\begin{proof}
To apply Proposition \ref{prop:1}, consider an information flow graph $G$ that enumerates all possible failure/repair patterns and all possible data collectors when the number of failures/repairs is bounded by $r$. We analyze the connectivity in the information flow graph to find the minimum repair bandwidth. Let the graph have $n$ initial nodes connected directly to the source and obtain $\alpha$ bits, while $r$ newcomers join the graph by connecting to $n-r$ complete nodes, obtaining $\beta$ bits from each. Any data collector DC that connects to a $k$-subset of ``out-nodes'' of $G$ must satisfy
\begin{align}\label{bound}
    C=\text{mincut(S,DC)}\geq \sum_{s=0}^{\sfrac{k}{r}}\min \{ (r\alpha_1 + (n-rs)\beta, r\alpha) \}.
\end{align}
First, we show that there exists an information flow graph $G^\prime$ where the bound (\ref{bound}) is matched with equality. The active nodes before the $s$-th repair round is triggered are labelled as $\mathcal{R}_s^{act}=\{(s-1)n+1, \ldots , sn\}$. Consider that in the $s$-th  repair round, the nodes $(s-1)n+(s-1)r+1, \ldots ,(s-1)n+sr$ are repaired. Denote the set of labels of the newcomers in the $s$-th round by $\mathcal{R}_s^{new}=\{ sn+(s-1)r+1, \ldots, sn+sr \}$, which represent the repaired nodes. The complete nodes are copied into the next round and labeled as $\mathcal{R}_s^{comp}=\{\{sn+i\}: i\in [n]\setminus \mathcal{R}_s^{new}\}$. The newcomers and the copied complete nodes together form the set of active nodes for the next repair round, i.e., $\mathcal{R}_{s+1}^{act}= \mathcal{R}_s^{new} \cup \mathcal{R}_s^{comp}$. Consider that the data collector connects to the nodes with indices $\mathcal{R}_s^{new}, s=[\sfrac{k}{r}]$. Consider a cut $(U,\bar{U})$ defined as follows. For the $s$-th repair round, if $r\alpha \leq r\alpha_1 + (n-sr)\beta$, then we include the nodes $x_{out}^{\mathcal{R}_s}$ in $\bar{U}$; otherwise, we include $x_{in}^{\mathcal{R}_s},x_{out}^{\mathcal{R}_s}$ in $\bar{U}$. Then this cut $(U,\bar{U})$ achieves (\ref{bound}) with equality. \\

Second, we show that any information flow graph has at least the minimum cut capacity of (\ref{bound}). There is a topological order of the nodes in an information flow graph by which any node $\nu_i$ having incoming edges only from nodes in $\bar{U}$, also belongs to $\bar{U}$, and an edge from $\nu_i$ to $\nu_j$ implies $i<j$. The min-cut could either be behind the helper nodes, i.e., all helper nodes (and the newcomers, whose parent nodes are the helper nodes) lie in $\bar{U}$, as illustrated in Fig. \ref{fig:flow1} by the cut $\chi_1$; or the min-cut could be after the helper nodes, i.e., all helper nodes (and the in-vertices of the newcomers) lie in $U$ while the out-vertices of the newcomers lie in $\bar{U}$. In the $s$-th round, the capacity contribution of the first $(s-1)r$ complete nodes will be zero, since their parent nodes from the previous rounds are in $\bar{U}$. If the min-cut is behind the helper node of even one node being repaired, it is easy to see that the min-cut will be behind the helper nodes of the remaining nodes being repaired, since the capacity contribution from the remaining nodes being repaired must be chosen as the minimum of $\alpha_1$ and $\alpha$, and $\alpha_1 \leq \alpha$ by definition. Thus it is shown that the min-cut cuts a repair round either behind the helper nodes or after the helper nodes of all the nodes being repaired. Thus the expression for the min-cut capacity can be determined in the following way. 

The contribution of the first repair round to the minimum cut capacity is $\min \{ r\alpha_1 + (n-r)\beta , r\alpha \}$, where the first term denotes the capacity contribution when the min-cut lies behind the helper nodes, and the second term denotes the contribution when the min-cut lies after the helper nodes.\\
Consider the second repair round. Following a similar line of reasoning, we analyze the different cases:

\begin{itemize}
    \item If the min-cut lies behind the helper nodes, since the capacity contribution from the nodes in $\mathcal{R}_1^{new}$ is zero, the total capacity contribution from the second round will be $r\alpha_1 + (n-2r)\beta$.
    \item If the min-cut lies after the helper nodes, then the capacity contribution will be $r\alpha$. 
\end{itemize}

The number of repair rounds for which we follow a similar procedure is $\sfrac{k}{r}$, since in each repair round there is a new set of $r$ nodes being repaired, and it is sufficient for the min-cut to cut out the $k$ nodes connected to the data collector. \\  
Thus the total capacity of the min-cut is given by Equation (\ref{bound}).

From Proposition \ref{prop:1}, the cut-capacity must be greater than the file size to ensure that the data collector is able to retrieve the file from any $k$ nodes. Therefore, the following must be satisfied for guaranteed file availability.
\begin{align} \label{sufficiency}
    \sum_{s=0}^{\sfrac{k}{r}}\min \{ (r\alpha_1 + (n-rs)\beta, r\alpha) \} \geq M
\end{align}

We are interested in characterizing the achievable trade-offs between the storage $\alpha$ and the repair bandwidth $(n-r)\beta$ for some given $(n,k,\rho)$.  \\
If $r\alpha \leq r\alpha_1 + (n-k)\beta$, then the min-cut passes behind the helper nodes in each of the $\sfrac{k}{r}$ repair rounds; if $ r\alpha_1 + (n-k)\beta \leq r\alpha \leq r\alpha_1 + (n-k+r)\beta$, then the min-cut passes behind the helper nodes for the first $\sfrac{k}{r}-1$ repair rounds, but passes on the right side of the helper nodes in the $\sfrac{k}{r}$-th repair round; and so on. In general, if $r\alpha_1 + (n-k+(s-1)r)\beta \leq r\alpha \leq r\alpha_1 + (n-k+sr)\beta, s \in [\sfrac{k}{r}]$, then the min-cut will cut behind the helper nodes for $s$ repair rounds, and cut after the helper nodes in the remaining $\sfrac{k}{r}-s$ repair rounds. \\

Let $b_s \triangleq \frac{\frac{n-k}{r}+s}{1-\rho}\beta, s=[0:\sfrac{k}{r}]$. The capacity of the min-cut is a piecewise-linear function of $\alpha$ given by
\begin{align}
    C(\alpha)&=\left\{
                \begin{array}{ll}
                k\alpha \hspace{170pt} \alpha \in (0,b_0]\\
                (k-r)\alpha + \left( r\alpha_1 + (n-k)\beta\right) \hspace{57pt} \alpha \in (b_0,b_1]\\
                \vdots \\
                r\alpha + \sum_{i=0}^{\sfrac{k}{r}-2} \left(  r\alpha_1 + (n-k+ir)\beta\right) \hspace{28pt} \alpha \in (b_{\sfrac{k}{r}-2}, b_{\sfrac{k}{r}-1}]\\
                 \sum_{i=0}^{\sfrac{k}{r}-1} \left(  r\alpha_1 + (n-k+ir)\beta\right) \hspace{51.5pt} \alpha \in (b_{\sfrac{k}{r}-1}, \infty]
                \end{array}
                \right. \\
            &= \left\{
                \begin{array}{ll}
                k\alpha \hspace{168pt} \alpha \in (0,b_0]\\
                (k-ir(1-\rho))\alpha +(1-\rho) \sum_{j=0}^{i-1} rb_j \hspace{26pt} \alpha \in (b_{i-1},b_i], i=1,2,\ldots, \sfrac{k}{r}-1 \\
                (1-\rho)\sum_{j=0}^{\sfrac{k}{r}-1} rb_j \hspace{102pt} \alpha \in (b_{\sfrac{k}{r}-1},\infty]
                \end{array}
                \right.    
\end{align}
Note that $C(\alpha)$ is a strictly increasing function. To find the minimum $\alpha$ for a given repair bandwidth $\gamma=(n-r)\beta$ such that $C(\alpha)\geq M$, we let $\alpha^* = C^{-1}(M)$ to obtain
\begin{align}
    \alpha^*&=\left\{
                \begin{array}{ll}
                \frac{M}{k} \hspace{45pt}  M\in (0,kb_0] \\
                \frac{M - g(i)\gamma}{k-ir(1-\rho)} \hspace{15pt} M \in \left[(k-(i-1)(1-\rho)r)b_{i-1}+(1-\rho)\sum_{j=0}^{i-2}rb_j , (k-i(1-\rho)r)b_{i-1}+(1-\rho)\sum_{j=0}^{i-1}rb_j\right] 
                \end{array}
                \right. \\
             &=\left\{
                \begin{array}{ll}
                  \frac{M}{k} \hspace{16.5mm} \gamma \in \left[ f(0), \infty \right) \\
                  \frac{M - g(i)\gamma}{k-ir(1-\rho)} \hspace{6mm} \gamma \in \left[ f(i), f(i-1)\right] 
                \end{array}
              \right. 
\end{align}
\end{proof}

\begin{corollary} \label{corollary1}
The minimum storage point is achieved by the pair $(\alpha_{MSR}, \gamma_{MSR})=\left(\frac{M}{k}, \frac{Mr(n-r)(1-\rho)}{k(n-k)}\right)$.
\end{corollary}

\begin{corollary} \label{corollary2}
The minimum repair bandwidth point is achieved by the pair $(\alpha_{MBR}, \gamma_{MBR}^*)=\left( \frac{M-g'\gamma^*_{MBR}}{k\rho + r(1-\rho)}, \frac{2Mr(n-r)(1-\rho)}{k(2n-k(1-\rho)-r(1+\rho))}\right)$ where $g'=\frac{1}{2}\frac{(k-r)(2n-k-2r)}{n-r}$.
\end{corollary}

Minimum-storage regenerating (MSR) and minimum-bandwidth regenerating (MBR) codes attain the points in Corollary \ref{corollary1} and Corollary \ref{corollary2}, respectively. 

\begin{remark}
For $\rho=0$ and $r=1$ , i.e., complete failure of exactly one node, the model is equivalent to that in \cite{5550492}, and the trade-off curve from Theorem \ref{theorem 1} coincides with the trade-off curve in \cite{5550492}. Similarly, for $\rho=0 \text{ and } r>1$, i.e., multiple complete failures, the trade-off curve from Theorem \ref{theorem 1} coincides with the trade-off curve in \cite{7459908}. 
\end{remark}

\begin{theorem}\label{theorem 2}
In the same context as in Theorem \ref{theorem 1}, if $r$ does not divide $k$, let $p\triangleq \lfloor \sfrac{k}{r} \rfloor$ and $k_0\triangleq pr$. Assume $\frac{n-k_0+(z-1)r}{r}\leq \frac{n-k_0 -r}{k-k_0}\leq \frac{n-k_0+zr}{r}$, for some $z\in [p-2]$, or $0 \leq \frac{n-k_0 -r}{k-k_0}\leq \frac{n-k_0}{r}$ for $z=0$. Also define $k^\prime \triangleq k\rho + (1-\rho)k_0$. Then the threshold function $\alpha^{*}(n,k,\gamma, r, \rho)$ is given by:
\begin{align}
\alpha^{*}=\left\{
                \begin{array}{ll}
                  \frac{M - g(i)\gamma}{k-ir(1-\rho)} \hspace{20.2mm} \gamma \in [f(i), f(i-1)] ,\\
                 \hspace{34mm}0\leq i \leq z-1 \\ \\
                  \frac{M - g(z)\gamma}{k-zr(1-\rho)} \hspace{19.5mm} \gamma \in \left[ f^\prime, f(z-1)\right] \\ \\
              \frac{M - [g(z)+\frac{n-k_0-r}{n-r}]\gamma}{k^\prime -zr(1-\rho)} \hspace{6.1mm} \gamma \in \left[ f(z), f^\prime \right] \\ \\
              \frac{M - [g(i)+\frac{n-k_0-r}{n-r}]\gamma}{k^\prime -ir(1-\rho)} \hspace{7mm} i\geq z+1, \\
           \hspace{34.8mm}   \gamma \in \left[ f(i), f(i-1) \right]
                \end{array}
              \right.
\end{align}
where $i={0,1,\ldots , \frac{k}{r}-1}$, and $f, g \text{ and } f^\prime$ are defined as
\begin{align}
f(i)&\triangleq \left\{
\begin{array}{ll}
\infty \hspace{44.7mm} i=-1\\ \\
\frac{2M(1-\rho)(n-r)}{(2k-r(i+1)(1-\rho))i+\frac{2k}{r}(n-k)} \hspace{10.3mm} i\leq z-1 \\ \\
\frac{2M(1-\rho)(n-r)}{(2k^\prime-r(i+1)(1-\rho))i+\frac{2k^\prime (n-k_0)}{r}+ n-k_0-r} \\ \hspace{48.3mm} i\geq z
\end{array}
\right.
\end{align}
\begin{align}
g(i)&\triangleq \frac{1}{2}\left(2n-2k_0-r+ir\right)\frac{ir}{n-r} \\ 
f^\prime &\triangleq 
\left\{
\begin{array}{ll}
\frac{2M(n-r)}{\left[\frac{2(n-k_0)(k-k_0-r)+2r^2}{k-k_0}+(z-1)r\right]z + \frac{2k(n-k_0-r)}{(k-k_0)(1-\rho)}} \\
\hspace{47.3mm} \text{if}\  z>0\\ 
\frac{M(k-k_0)(n-r)(1-\rho)}{k^\prime (n-k_0-r)} \hspace{20mm} \text{if}\  z=0.
\end{array}
\right.
\end{align}

\end{theorem}


\section{Special Cases} \label{examples}
\begin{example} \label{example 1} \normalfont
Consider a network with the following parameters: $n=4$, $k=3$, $r=2$ and $\alpha_1 = \frac{\alpha}{2}$ (Fig. \ref{fig:flow2}). According to Eq. (6) in \cite{7000553}, the MSR point is achieved by the pair $(\alpha, \gamma)=(\frac{M}{3}, M)$. However, by constructing the information flow graph as shown in Fig. \ref{fig:flow2} and finding the minimum cutset, we find that the storage point of $\frac{M}{3}$ is not feasible, and instead the storage point $\alpha = \frac{2M}{5}$ is feasible, and the corresponding repair bandwidth point is $\gamma = \frac{2M}{5}$. This is verified by Theorem \ref{theorem 2}, where $z=0$ for this example, and therefore $f'=f(z-1)=\infty$. Therefore $\alpha = \frac{M}{k}$ is not possible. Interestingly, when we substitute the value of $\alpha=\frac{2M}{5}$ into Eq. (6) of \cite{7000553}, the repair bandwidth in \cite{7000553} coincides with ours. 
\end{example}
\begin{example} \normalfont
Consider the parameters: $n=4, k=2, r=1, \alpha_1 = \frac{\alpha}{2}$. According to Eq. (6) in \cite{7000553}, the MSR point is achieved by the pair $(\alpha, \gamma)=(\frac{M}{2}, M)$. Our approach achieves $\gamma = \frac{3M}{8}$, which is achievable using linear network codes in $GF(q)$, assuming $q$ is large enough. We demonstrate the achievability of the bound for this example with the following code construction.
\paragraph{Code Construction}
Split the file into $8$ non-overlapping packets of size $\frac{M}{8}$ bits, denoted by $\mathbf{W} = (\mathbf{w}_1, \mathbf{w}_2, \ldots, \mathbf{w}_8)$. Encode these 8 packets with a $(16, 8)$ MDS code in $GF(q)$. 
The coded packets are obtained by simple matrix multiplication of the generator matrix $\mathbf{G}$ for the $(16,8)$ MDS code with $\mathbf{W}$, i.e. $\mathbf{P}=\mathbf{W} \times \mathbf{G}$. Let the $j-$th column of the matrix $\mathbf{P}$ represent the $j-$th coded packet $\mathbf{p}_j$ where $j=1, \ldots, 16$. Node $i$ where $i=1, \ldots, 4$ stores the packets corresponding to the columns from $4i-3$ to $4i$. It holds that $Rank(\mathbf{P})=8$, therefore every submatrix $\mathbf{P}'$ of $\mathbf{P}$ consisting of $8$ columns has a full rank. This implies that for any $8$- dimensional vector $\mathbf{b}$, there exist solutions $\mathbf{y}_i, i=1,2$ for the equations
\begin{align}\label{eq:span}
\mathbf{P}'_i\mathbf{y}_i=\mathbf{b} \ \ \ \ \ \ \  i=1,2
\end{align}
Without loss of generality, suppose node 1 loses two of its four packets. Thus $\alpha = \frac{M}{2}, \alpha_1=\frac{\alpha}{2}=\frac{M}{4}$. In the rest of this subsection, we describe how the transmissions from nodes $2,3$ and $4$ may be designed so that node $1$ recovers its lost packets $\mathbf{p}_1$ and $\mathbf{p}_2$ from the received packets and its remaining packets $\mathbf{p}_3$ and $\mathbf{p}_4$. \\
Create the $8\times 8$ matrix $\mathbf{P}'_1$ from the first $8$ columns of $\mathbf{P}$, and $\mathbf{P}'_2$ from the last $8$ columns of $\mathbf{P}$.
Define an $8$-dimensional vector $\mathbf{y}_1\triangleq (y_1 \ldots y_4 \ 0 \ldots 0)^{T}$. From \eqref{eq:span}, we obtain $ \mathbf{b}_1=\mathbf{P'}_1 \mathbf{y}_1$. Find $\mathbf{y}_2=\mathbf{P'}{_2^{-1}}\mathbf{b}_1=\mathbf{P'}{_2^{-1}}\mathbf{P'}_1\mathbf{y}_1$. Representing $\mathbf{y}_2$ as $(\mathbf{y}^1_2 \  \mathbf{y}^2_2)^{T}$, where $\mathbf{y}_2^1$ and $\mathbf{y}_2^2$ are the vectors containing the first four and last four elements of $\mathbf{y}_2$ respectively, then we define
\begin{align*}
\mathbf{x}_3&\triangleq \mathbf{y}^1_2 * (\mathbf{p}_9 \ldots \mathbf{p}_{12})^{T}\\
\mathbf{x}_4& \triangleq \mathbf{y}^2_2 * (\mathbf{p}_{13} \ldots \mathbf{p}_{16})^{T}.
\end{align*}
Next, set $\mathbf{b}_2=\gamma_{3}\mathbf{x}_3 + \gamma_4 \mathbf{x}_4$, for some arbitrary constants $\gamma_3, \gamma_4 \in GF(q)\  s.t. \gamma_3 \neq \gamma_4$. Solve the equation $\mathbf{P}'_1 \mathbf{y}_3=\mathbf{b}_2$ to obtain $\mathbf{y}_3$. Write $\mathbf{y}_3$ as $(\mathbf{y}^1_3 \  \mathbf{y}^2_3)^{T}$, where $\mathbf{y}_3^1$ and $\mathbf{y}_3^2$ are the vectors containing the first and last four elements of $\mathbf{y}_3$ respectively. Then we define
\begin{align*}
\mathbf{x}_2 \triangleq \mathbf{y}^2_3* (\mathbf{p}_5 \ldots \mathbf{p}_8)^{T}.
\end{align*}
Therefore, we obtain the following two linear equations:

\begin{align}
\nonumber \mathbf{P}'_1 \mathbf{y}_1&=\mathbf{x}_3 + \mathbf{x}_4\\
\mathbf{P}'_1 \mathbf{y}_3&=\gamma_{3}\mathbf{x}_3 + \gamma_4 \mathbf{x}_4. \label{eq:recovery}
\end{align}

When node $1$ loses any two packets, nodes $2,3$ and $4$ transmit the packets $\mathbf{x}_2, \mathbf{x}_3$ and $\mathbf{x}_4$ respectively. Node $1$ can solve for the two unknown packets from the set of equations \eqref{eq:recovery}. Since nodes $2,3$ and $4$ transmit one packet each of size $\frac{M}{8}$ bits, $\gamma$ is equal to $\frac{3M}{8}$.
\end{example}

\begin{example} \normalfont
Consider a network with the following parameters: $n=4$, $k=2$, $r=2$ and $\alpha_1 = \frac{\alpha}{2}$. Reference \cite{7000553} illustrates an achievable scheme for this example. It can be verified from \eqref{theorem 1} that the scheme achieves the optimal repair bandwidth , i.e., the pair is $(\alpha, \gamma)=(\frac{M}{2}, \frac{M}{2})$.
\end{example}
\section{Results and discussion}

\begin{figure} 
    \centering
    \includegraphics[scale=0.6]{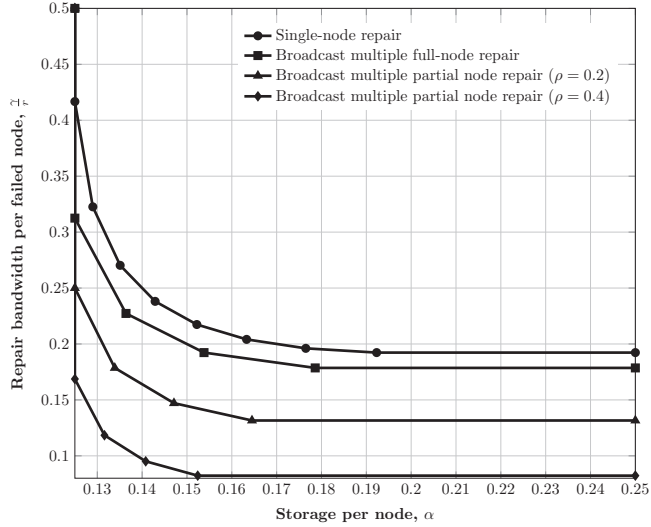}
    \caption{Trade-off curve between the repair bandwidth and storage, $M=1,k=8$ and $10$ helper nodes. For single node failure $r=1$ and for multiple node failures $r=2$.} \label{tradeoff}
\end{figure}

In Fig. \ref{tradeoff}, we plot the storage-repair bandwidth (per failed node) trade-off for single-node repair \cite{5550492}, broadcast repair of multiple full node failures \cite{7459908}, and partial repair of multiple nodes. Fig. \ref{tradeoff} illustrates that utilizing the remaining portion of data that is not lost on a failed node reduces the repair bandwidth significantly. We observe that the repair bandwidth reduces quickly for small values of storage capacity $\alpha$, and saturates at a fixed value beyond a particular threshold value of $\alpha$. That point is the MBR point. The threshold value of $\alpha$ becomes smaller for larger values of $\rho$. There is another threshold value of $\alpha$ below which repair with a finite repair bandwidth is not feasible and corresponds to the MSR point. 

\section{Conclusion} \label{conc}
In this paper, we consider the problem of repair of partial failures of multiple nodes by broadcast transmissions in a wireless medium. For this setting we derive the optimal storage-repair bandwidth trade-off curve by constructing a time-evolving information flow graph to represent the evolution of the system with time, and finding the minimum cutset across all failure combinations. Our results show that some pairs of storage and repair bandwidth values from related literature are not feasible generally. It has been shown in previous literature that compared to the single node repair, repairing multiple nodes simultaneously and exploiting the broadcast nature of the medium reduces the repair bandwidth per failed node. We illustrate that the optimal repair bandwidth is reduced even further by using the remaining content in the cache nodes that experience partial failure. Additionally, we demonstrate the achievability of the derived bounds for a special case with an explicit code construction. We can deduce that designing storage nodes as clusters of independent storage units, so that failure of a few storage units (partial failure) can be repaired using the scheme described in this paper, is more efficient.

\end{document}